\documentclass{llncs}
\usepackage{version}
\usepackage{pgabc}

\title{Axioms for Behavioural Congruence \\
       of Single-Pass Instruction Sequences}
\author{J.A. Bergstra \and C.A. Middelburg}
\institute{Informatics Institute, Faculty of Science, University of
           Amsterdam \\
           Science Park~904, 1098~XH Amsterdam, the Netherlands \\
           \email{J.A.Bergstra@uva.nl,C.A.Middelburg@uva.nl}}

\begin{document}
\maketitle

\begin{abstract}
In program algebra, an algebraic theory of single-pass instruction 
sequences, three congruences on instruction sequences are paid attention 
to: instruction sequence congruence, structural congruence, and 
behavioural congruence.
Sound and complete axiom systems for the first two congruences were 
already given in early papers on program algebra.
The current paper is the first one that is concerned with an axiom 
system for the third congruence.
The presented axiom system is especially notable for its axioms that 
have to do with forward jump instructions.
\end{abstract}

\section{Introduction}
\label{sect-intro}

Program algebra, an algebraic theory of single-pass instruction 
sequences, was first presented in~\cite{BL02a} as the basis of an 
approach to programming language semantics.
Various issues, including issues relating to programming language 
expressive\-ness, computability, computational complexity, algorithm 
efficiency, algorithmic equivalence of programs, program verification, 
program perform\-ance, program compactness, and program parallelization, 
have been studied in the setting of program algebra since then.
An overview of all the work done to date and some open questions 
originating from it can be found at~\cite{SiteIS}.
Three congruences on instruction sequences were introduced 
in~\cite{BL02a}: instruction sequence congruence, structural congruence 
and behavioural congruence.
Sound and complete axiom systems for instruction sequence congruence 
and structural congruence were already provided in~\cite{BL02a}, but an 
axiom system for behavioural congruence has never been provided.
This paper is concerned with an axiom system for behavioural 
congruence.

Program algebra is parameterized by a set of uninterpreted basic 
instructions.
In applications of program algebra, this set is instantiated by a set of 
interpreted basic instructions.
In the case of most issues that have been studied in the setting of 
program algebra, the interpreted basic instructions are instructions to 
set and get the content of Boolean registers.
In the case of a few issues, the interpreted basic instructions are 
other instructions, e.g.\ instructions to manipulate the content of 
counters or instructions to manipulate the content of Turing machine 
tapes (see e.g.~\cite{BM09k}).

In the uninstantiated case, behavioural congruence is the coarsest 
congruence respecting the behaviour produced by instruction sequences 
under execution that is possible with uninterpreted basic instructions.
In the instantiated cases, behavioural congruence is the coarsest 
congruence respecting the behaviour produced by instruction sequences 
under execution that is possible taking the intended interpretation of 
the basic instructions into account.
In this paper, an emphasis is laid on the uninstantiated case.
Yet attention is paid to the instantiation in which all possible 
instructions for Boolean registers are taken as basic instructions.

The single-pass instruction sequences considered in program algebra are
non-empty, finite or eventually periodic infinite instruction sequences.
In this paper, the soundness question, i.e.\ the question whether 
derivable equality implies behavioural congruence, is fully answered in 
the affirmative.
However, the completeness question, i.e.\ the question whether 
behavioural congruence implies derivable equality, is answered in the 
affirmative only for the restriction to finite instruction sequences
because of problems in mastering the intricacy of a completeness proof 
for the unrestricted case.   

In~\cite{BL02a}, basic thread algebra, an algebraic theory of 
mathematical objects that model in a direct way the behaviours produced 
by instruction sequences under execution, was introduced to describe 
which behaviours are produced by the instruction sequences considered 
in program algebra.%
\footnote
{In~\cite{BL02a}, basic thread algebra is introduced under the name
 basic polarized process algebra.}
It is rather awkward to describe and analyze the behaviours of this 
kind using algebraic theories of processes such as 
ACP~\cite{BW90,BK84b}, CCS~\cite{HM85,Mil89} and CSP~\cite{BHR84,Hoa85}. 
However, the objects considered in basic thread algebra can be viewed as 
representations of processes as considered in ACP 
(see e.g.~\cite{BM11c}).
Basic thread algebra is parameterized by a set of uninterpreted basic 
actions and, when it is used for describing the behaviours produced by 
instruction sequences under execution, basic instructions are taken as 
basic actions.
Like in~\cite{BL02a}, basic thread algebra will be used in this paper 
for describing the behaviours produced by the instruction sequences 
considered in program algebra and to define the notion of behavioural 
congruence of instruction sequences.

This paper is organized as follows.
First, we introduce a version of program algebra with axioms for 
instruction sequence congruence, structural congruence, and behavioural 
congruence (Section~\ref{sect-PGA}).
Next, we present the preliminaries on basic thread algebra that are 
needed in the rest of the paper (Section~\ref{sect-BTA}).
After that, we describe which behaviours are produced by instruction 
sequences under execution and define a notion of behavioural congruence 
for instruction sequences (Section~\ref{sect-TE-BC}).
Then, we go into the soundness and completeness of the presented axiom 
system with respect to the defined notion of behavioural congruence 
(Section~\ref{sect-PGA-BC}).
Following this, we look at the instantiation of program algebra in which 
all possible instructions for Boolean registers are taken as basic 
instructions (Section~\ref{sect-PGAbr}).
Finally, we make some concluding remarks (Section~\ref{sect-concl}).

The following should be mentioned in advance.
The set $\Bool$ of Boolean values is a set with two elements whose 
intended interpretations are the truth values \emph{false} and 
\emph{true}.
As is common practice, we represent the elements of $\Bool$ by the bits 
$0$ and $1$.

This paper draws somewhat from the preliminaries of earlier papers that 
built on program algebra and basic thread algebra.
The most recent one of the papers in question is~\cite{BM15a}.

\section{Program Algebra for Behavioural Congruence}
\label{sect-PGA}

In this section, we present \PGAbc.
\PGAbc\ is a version of \PGA\ (ProGram Algebra) with, in addition to the 
usual axioms for instruction sequence congruence and structural 
congruence, axioms for behavioural congruence.

The instruction sequences considered in \PGAbc\ are single-pass 
instruction sequences of a particular kind.%
\footnote
{The instruction sequences concerned are single-pass in the sense that 
they are instruction sequences of which each instruction is executed at 
most once and can be dropped after it has been executed or jumped over.}
It is assumed that a fixed but arbitrary set $\BInstr$ of 
\emph{basic instructions} has been given.
$\BInstr$ is the basis for the set of instructions that may occur in 
the instruction sequences considered in \PGAbc.
The intuition is that the execution of a basic instruction may modify a
state and must produce a Boolean value as reply at its completion.
The actual reply may be state-dependent.

The set of instructions of which the instruction sequences considered 
in \PGAbc\ are composed is the set that consists of the following 
elements:
\begin{itemize}
\item
for each $a \in \BInstr$, a \emph{plain basic instruction} $a$;
\item
for each $a \in \BInstr$, a \emph{positive test instruction} $\ptst{a}$;
\item
for each $a \in \BInstr$, a \emph{negative test instruction} $\ntst{a}$;
\item
for each $l \in \Nat$, a \emph{forward jump instruction} $\fjmp{l}$;
\item
a \emph{termination instruction} $\halt$.
\end{itemize}
We write $\PInstr$ for this set.
The elements from this set are called \emph{primitive instructions}.

Primitive instructions are the elements of the instruction sequences 
considered in \PGAbc.
On execution of such an instruction sequence, these primitive 
instructions have the following effects:
\begin{itemize}
\item
the effect of a positive test instruction $\ptst{a}$ is that basic
instruction $a$ is executed and execution proceeds with the next
primitive instruction if $\True$ is produced and otherwise the next
primitive instruction is skipped and execution proceeds with the
primitive instruction following the skipped one --- if there is no
primitive instruction to proceed with,
inaction occurs;
\item
the effect of a negative test instruction $\ntst{a}$ is the same as
the effect of $\ptst{a}$, but with the role of the value produced
reversed;
\item
the effect of a plain basic instruction $a$ is the same as the effect
of $\ptst{a}$, \linebreak[2] but execution always proceeds as if $\True$ 
is produced;
\item
the effect of a forward jump instruction $\fjmp{l}$ is that execution
proceeds with the $l$th next primitive instruction --- if $l$ equals $0$ 
or there is no primitive instruction to proceed with, inaction occurs;
\item
the effect of the termination instruction $\halt$ is that execution 
terminates.
\end{itemize}
Inaction occurs if no more basic instructions are executed, but 
execution does not terminate.

\PGAbc\ has one sort: the sort $\InSeq$ of \emph{instruction sequences}. 
We make this sort explicit to anticipate the need for many-sortedness
later on.
To build terms of sort $\InSeq$, \PGAbc\ has the following constants and 
operators:
\begin{itemize}
\item
for each $u \in \PInstr$, 
the \emph{instruction} constant $\const{u}{\InSeq}$\,;
\item
the binary \emph{concatenation} operator 
$\funct{\ph \conc \ph}{\InSeq \x \InSeq}{\InSeq}$\,;
\item
the unary \emph{repetition} operator 
$\funct{\ph\rep}{\InSeq}{\InSeq}$\,.
\end{itemize}
Terms of sort $\InSeq$ are built as usual in the one-sorted case.
We assume that there are infinitely many variables of sort $\InSeq$, 
including $X,Y,Z$.
We use infix notation for concatenation and postfix notation for
repetition.

A \PGAbc\ term in which the repetition operator does not occur is called 
a \emph{repetition-free} \PGAbc\ term.

One way of thinking about closed \PGAbc\ terms is that they represent 
non-empty, finite or eventually periodic infinite sequences of 
primitive instructions.%
\footnote
{An eventually periodic infinite sequence is an infinite sequence with
 only finitely many distinct suffixes.}
The instruction sequence represented by a closed term of the form
$t \conc t'$ is the instruction sequence represented by $t$
concatenated with the instruction sequence represented by $t'$.
The instruction sequence represented by a closed term of the form 
$t\rep$ is the instruction sequence represented by $t$ concatenated 
infinitely many times with itself.
A closed \PGAbc\ term represents a finite instruction sequence if and 
only if it is a closed repetition-free \PGAbc\ term.

In this paper, closed \PGAbc\ terms are considered equal if the 
instruction sequences that they represent can always take each other's 
place in an instruction sequence in the sense that the behaviour 
produced under execution remains the same irrespective of the 
interpretation of the instructions from $\BInstr$.
In other words, equality of closed terms stands in \PGAbc\ for a kind of
behavioural congruence of the represented instruction sequences.
The kind of behavioural congruence in question will be made precise in 
Section~\ref{sect-TE-BC}.

The axioms of \PGAbc\ are given in Table~\ref{axioms-PGAbc}.%
\begin{table}[!p]
\caption{Axioms of \PGAbc}
\label{axioms-PGAbc}
\begin{seqntbl}
\renewcommand{\arraystretch}{1.37}
\begin{axcol}
(X \conc Y) \conc Z = X \conc (Y \conc Z)             & \axiom{PGA1}  \\
(X^n)\rep = X\rep                                     & \axiom{PGA2}  \\
X\rep \conc Y = X\rep                                 & \axiom{PGA3}  \\
(X \conc Y)\rep = X \conc (Y \conc X)\rep             & \axiom{PGA4} 
\eqnsep
\fjmp{k{+}1} \conc u_1 \conc \ldots \conc u_k \conc \fjmp{0} =
\fjmp{0} \conc u_1 \conc \ldots \conc u_k \conc \fjmp{0} 
                                                      & \axiom{PGA5}  \\
\fjmp{k{+}1} \conc u_1 \conc \ldots \conc u_k \conc \fjmp{l} =
\fjmp{l{+}k{+}1} \conc u_1 \conc \ldots \conc u_k \conc \fjmp{l}
                                                      & \axiom{PGA6}  \\
(\fjmp{l{+}k{+}1} \conc u_1 \conc \ldots \conc u_k)\rep =
(\fjmp{l} \conc u_1 \conc \ldots \conc u_k)\rep       & \axiom{PGA7}  \\
\fjmp{l{+}k{+}k'{+}2} \conc u_1 \conc \ldots \conc u_k \conc
(v_1 \conc \ldots \conc v_{k'{+}1})\rep = {} \\ \phantom{{}{+}k'}
\fjmp{l{+}k{+}1} \conc u_1 \conc \ldots \conc u_k \conc
(v_1 \conc \ldots \conc v_{k'{+}1})\rep               & \axiom{PGA8} 
\eqnsep
\ptst{a} \conc \fjmp{0} \conc \fjmp{0} = 
a \conc \fjmp{0} \conc \fjmp{0}                       & \axiom{PGA9} \\
\ntst{a} \conc \fjmp{0} \conc \fjmp{0} = 
a \conc \fjmp{0} \conc \fjmp{0}                       & \axiom{PGA10} \\
\ptst{a} \conc \fjmp{1} = a \conc \fjmp{1}            & \axiom{PGA11} \\
\ntst{a} \conc \fjmp{1} = a \conc \fjmp{1}            & \axiom{PGA12} \\
\ptst{a} \conc \fjmp{l{+}2} \conc \fjmp{l{+}1} = 
a \conc \fjmp{l{+}2} \conc \fjmp{l{+}1}               & \axiom{PGA13} \\
\ntst{a} \conc \fjmp{l{+}2} \conc \fjmp{l{+}1} = 
a \conc \fjmp{l{+}2} \conc \fjmp{l{+}1}               & \axiom{PGA14} \\
\ptst{a} \conc \halt \conc \halt = 
a \conc \halt \conc \halt                             & \axiom{PGA15} \\
\ntst{a} \conc \halt \conc \halt = 
a \conc \halt \conc \halt                             & \axiom{PGA16} \\
\ptst{a} \conc u\rep = a \conc u\rep                  & \axiom{PGA17} \\
\ntst{a} \conc u\rep = a \conc u\rep                  & \axiom{PGA18} \\
\fjmp{k{+}3} \conc \fjmp{k{+}3} \conc \fjmp{k{+}3} \conc 
u_1 \conc \ldots \conc u_k \conc \ptst{a} = {} 
\ptst{a} \conc \fjmp{k{+}3} \conc \fjmp{k{+}3} \conc 
u_1 \conc \ldots \conc u_k \conc \ptst{a}             & \axiom{PGA19} \\
\fjmp{k{+}3} \conc \fjmp{k{+}3} \conc \fjmp{k{+}3} \conc 
u_1 \conc \ldots \conc u_k \conc \ntst{a} = {} 
\ntst{a} \conc \fjmp{k{+}3} \conc \fjmp{k{+}3} \conc 
u_1 \conc \ldots \conc u_k \conc \ntst{a}             & \axiom{PGA20} \\
\fjmp{k{+}2} \conc \fjmp{k{+}2} \conc 
u_1 \conc \ldots \conc u_k \conc a = 
a \conc \fjmp{k{+}2} \conc u_1 \conc \ldots \conc u_k \conc a
                                                      & \axiom{PGA21} \\
\fjmp{k{+}k'{+}4} \conc u_1 \conc \ldots \conc u_k \conc 
\ptst{a} \conc \fjmp{k'{+}3} \conc \fjmp{k'{+}3} \conc 
v_1 \conc \ldots \conc v_{k'} \conc \ptst{a} = {} \\ \phantom{{}{+}k'}
\fjmp{k{+}1} \conc u_1 \conc \ldots \conc u_k \conc 
\ptst{a} \conc \fjmp{k'{+}3} \conc \fjmp{k'{+}3} \conc 
v_1 \conc \ldots \conc v_{k'} \conc \ptst{a}          & \axiom{PGA22} \\
\fjmp{k{+}k'{+}4} \conc u_1 \conc \ldots \conc u_k \conc 
\ntst{a} \conc \fjmp{k'{+}3} \conc \fjmp{k'{+}3} \conc 
v_1 \conc \ldots \conc v_{k'} \conc \ntst{a} = {} \\ \phantom{{}{+}k'}
\fjmp{k{+}1} \conc u_1 \conc \ldots \conc u_k \conc 
\ntst{a} \conc \fjmp{k'{+}3} \conc \fjmp{k'{+}3} \conc 
v_1 \conc \ldots \conc v_{k'} \conc \ntst{a}          & \axiom{PGA23} \\
\fjmp{k{+}k'{+}3} \conc u_1 \conc \ldots \conc u_k \conc 
a \conc \fjmp{k'{+}2} \conc  
v_1 \conc \ldots \conc v_{k'} \conc a = {} \\ \phantom{{}{+}k'}
\fjmp{k{+}1} \conc u_1 \conc \ldots \conc u_k \conc 
a \conc \fjmp{k'{+}2} \conc  
v_1 \conc \ldots \conc v_{k'} \conc a                 & \axiom{PGA24} \\
\fjmp{k{+}1} \conc u_1 \conc \ldots \conc u_k \conc \halt =
\halt \conc u_1 \conc \ldots \conc u_k \conc \halt    & \axiom{PGA25} \\
\fjmp{k{+}1} \conc (u_1 \conc \ldots \conc u_k \conc u)\rep =
(u \conc u_1 \conc \ldots \conc u_k)\rep              & \axiom{PGA26} \\
(\fjmp{k{+}2} \conc \fjmp{k{+}1} \conc 
 u_1 \conc \ldots \conc u_k \conc \ptst{a})\rep = {} 
(a \conc \fjmp{k{+}1} \conc 
 u_1 \conc \ldots \conc u_k \conc a)\rep              & \axiom{PGA27} \\
(\fjmp{k{+}2} \conc \fjmp{k{+}1} \conc 
 u_1 \conc \ldots \conc u_k \conc \ntst{a})\rep = {} 
(a \conc \fjmp{k{+}1} \conc 
 u_1 \conc \ldots \conc u_k \conc a)\rep              & \axiom{PGA28} \\
(\fjmp{k{+}2} \conc \fjmp{k{+}1} \conc 
 u_1 \conc \ldots \conc u_k \conc a)\rep = {} 
(a \conc \fjmp{k{+}1} \conc 
 u_1 \conc \ldots \conc u_k \conc a)\rep              & \axiom{PGA29} \\
(u_1 \conc \ldots \conc u_{k{+}1})\rep = a\rep \\
\quad\;\;\; \mathrm{if},\; 
\mathrm{for\;all}\; i \in \set{1,\ldots,k{+}1},\;
 u_i \in \set{a,\ptst{a},\ntst{a}}\; \mathrm{or},\; 
 \mathrm{for\;some}\; l \in \set{1,\ldots,k},\;
\\ \phantom{\quad\;\;\; \mathrm{if},\; {}}  
  u_i \equiv \fjmp{l}\; \mathrm{and}\;
  u_{(i+l) \mathrm{mod} (k+1)} \in \set{a,\ptst{a},\ntst{a}}                   
                                                      & \axiom{PGA30} 
\end{axcol}
\end{seqntbl}
\end{table}
In this table, 
$n$ stands for an arbitrary natural number from $\Natpos$,%
\footnote
{We write $\Natpos$ for the set $\set{n \in \Nat \where n \geq 1}$ of
positive natural numbers.}
$u$, $u_1,\ldots,u_k$ and $v_1,\ldots,v_{k'+1}$ stand for arbitrary 
primitive instructions from $\PInstr$,
$k$, $k'$, and $l$ stand for arbitrary natural numbers from $\Nat$, and 
$a$ stands for an arbitrary basic instruction from $\BInstr$.
For each $n \in \Natpos$, the term $t^n$, where $t$ is a \PGAbc\ term, 
is defined by induction on $n$ as follows: $t^1 = t$, and 
$t^{n+1} = t \conc t^n$.

If $t = t'$ is derivable from PGA1--PGA4, then $t$ and $t'$ represent 
the same instruction sequence.
In this case, we say that the represented instruction sequences are 
\emph{instruction sequence congruent}.
We write \PGAisc\ for the algebraic theory whose sorts, constants and
operators are those of \PGAbc, but whose axioms are PGA1--PGA4.

The \emph{unfolding equation} $X\rep = X \conc X\rep$ is derivable from
the axioms of \PGAisc\ by first taking the instance of PGA2 in which 
$n = 2$, then applying PGA4, and finally applying the instance of PGA2 
in which $n = 2$ again.

A closed \PGAbc\ term is in \emph{first canonical form} if it is of the 
form $t$ or $t \conc {t'}\rep$, where $t$ and $t'$ are closed 
repetition-free \PGAbc\ terms.
The following proposition relates \PGAisc\ and first canonical forms.
\begin{proposition}
\label{prop-1CF}
For all closed \PGAbc\ terms $t$, there exists a closed \PGAbc\ term $t'$ 
that is in first canonical form such that $t = t'$ is derivable from
the axioms of \PGAisc.
\end{proposition}
\begin{proof}
The proof is analogous to the proof of Lemma~2.2 from~\cite{BM12b}.
\qed 
\end{proof}

If $t = t'$ is derivable from PGA1--PGA8, then $t$ and $t'$ represent 
the same instruction sequence after changing all chained jumps into 
single jumps and making all jumps ending in the repeating part as short 
as possible if they are eventually periodic infinite sequences.
In this case, we say that the represented instruction sequences are 
\emph{structurally congruent}.
We write \PGAsc\ for the algebraic theory whose sorts, constants and
operators are those of \PGAbc, but whose axioms are PGA1--PGA8.

A closed \PGAbc\ term $t$ \emph{has chained jumps} if there exists a 
closed \PGAbc\ term $t'$ such that $t = t'$ is derivable from the axioms
of \PGAisc\ and $t'$ contains a subterm of the form 
$\fjmp{n{+}1} \conc u_1 \conc \ldots \conc u_n \conc \fjmp{l}$.
A closed \PGAbc\ term $t$ that is in first canonical form 
\emph{has a repeating part} if it is of the form
$u_1 \conc \ldots \conc u_m \conc (v_1 \conc \ldots \conc v_k)\rep$.
A closed \PGAbc\ term $t$ of the form
$u_1 \conc \ldots \conc u_m \conc (v_1 \conc \ldots \conc v_k)\rep$
\emph{has shortest possible jumps ending in the repeating part} if:
(i)~for each $i \in [1,m]$ for which $u_i$ is of the form $\fjmp{l}$,
$l \leq k + m - i$;
(ii)~for each $j \in [1,k]$ for which $v_j$ is of the form $\fjmp{l}$,
$l \leq k - 1$.
A closed \PGAbc\ term is in \emph{second canonical form} if it is in first 
canonical form, does not have chained jumps, and has shortest
possible jumps ending in the repeating part if it has a repeating part.
The following proposition relates \PGAsc\ and second canonical forms.
\begin{proposition}
\label{prop-2CF}
For all closed \PGAbc\ terms $t$, there exists a closed \PGAbc\ term $t'$ 
that is in second canonical form such that $t = t'$ is derivable from 
the axioms of \PGAsc.
\end{proposition}
\begin{proof}
The proof is analogous to the proof of Lemma~2.3 from~\cite{BM12b}.
\qed 
\end{proof}

If $t = t'$ is derivable from PGA1--PGA30, then $t$ and $t'$ represent 
instruction sequences that can always take each other's place in an 
instruction sequence without affecting the behaviour produced under 
execution in an essential way.
In this case, we say that the represented instruction sequences are 
\emph{behaviourally congruent}.
In Section~\ref{sect-TE-BC}, we will use basic thread algebra to make 
precise which behaviours are produced by the represented instruction 
sequences under execution.

Axioms PGA1--PGA8 originate from~\cite{BL02a}.
Axioms PGA9--PGA30 are new and some of them did not come into the 
picture until we recently attempted to obtain a complete axiom system 
for behavioural congruence.

Henceforth, the instruction sequences of the kind considered in \PGAisc, 
\PGAsc, and \PGAbc\ are called \PGA\ instruction sequences.

\section{Basic Thread Algebra for Finite and Infinite Threads}
\label{sect-BTA}

In this section, we present an extension of \BTA\ (Basic Thread Algebra)
that reflects the idea that infinite threads are identical if their 
approximations up to any finite depth are identical.

\BTA\ is concerned with mathematical objects that model in a direct 
way the behaviours produced by \PGA\ instruction sequences under 
execution.
The objects in question are called threads.
A thread models a behaviour that consists of performing basic actions in 
a sequential fashion.
Upon performing a basic action, a reply from an execution environment
determines how the behaviour proceeds subsequently.
The basic instructions from $\BInstr$ are taken as basic actions.

\BTA\ has one sort: the sort $\Thr$ of \emph{threads}. 
We make this sort explicit to anticipate the need for many-sortedness
later on.
To build terms of sort $\Thr$, \BTA\ has the following constants and 
operators:
\begin{itemize}
\item
the \emph{inaction} constant $\const{\DeadEnd}{\Thr}$;
\item
the \emph{termination} constant $\const{\Stop}{\Thr}$;
\item
for each $a \in \BAct$, the binary \emph{postconditional composition} 
operator $\funct{\pcc{\ph}{a}{\nolinebreak\ph}}{\Thr \x \Thr}{\Thr}$.
\end{itemize}
Terms of sort $\Thr$ are built as usual in the one-sorted case. 
We assume that there are infinitely many variables of sort $\Thr$, 
including $x,y,z$.
We use infix notation for postconditional composition. 
We introduce \emph{basic action prefixing} as an abbreviation: 
$a \bapf t$, where $t$ is a \BTA\ term, abbreviates 
$\pcc{t}{a}{t}$.
We treat an expression of the form $a \bapf t$ and the \BTA\ term that
it abbreviates as syntactically the same.

Different closed \BTA\ terms are considered to represent different 
threads.
The thread represented by a closed term of the form $\pcc{t}{a}{t'}$
models the behaviour that will first perform $a$, and then proceed as 
the behaviour modeled by the thread represented by $t$ if the reply from 
the execution environment is $\True$ and proceed as the behaviour 
modeled by the thread represented by $t'$ if the reply from the 
execution environment is $\False$. 
The thread represented by $\Stop$ models the behaviour that will do no 
more than terminate and the thread represented by $\DeadEnd$ models the 
behaviour that will become inactive.

Closed \BTA\ terms are considered equal if they represent the same 
thread.
Equality of closed terms stands in \BTA\ for syntactic identity.
Therefore, \BTA\ has no axioms.

Each closed \BTA\ term represents a finite thread, i.e.\ a thread with 
a finite upper bound to the number of basic actions that it can perform.
Infinite threads, i.e.\ threads without a finite upper bound to the
number of basic actions that it can perform, can be defined by means of 
a set of recursion equations (see e.g.~\cite{BM09k}).
A regular thread is a finite or infinite thread that can only be in a 
finite number of states.
The behaviours produced by \PGA\ instruction sequences under execution 
are exactly the behaviours modeled by regular threads.

Two infinite threads are considered identical if their approximations up 
to any finite depth are identical.
The approximation up to depth $n$ of a thread models the behaviour that 
differs from the behaviour modeled by the thread in that it will become
inactive after it has performed $n$ actions unless it would terminate at
this point.
AIP (Approximation Induction Principle) is a conditional equation that
formalizes the above-mentioned view on infinite threads.
In AIP, the approximation up to depth $n$ is phrased in terms of the
unary \emph{projection} operator $\funct{\proj{n}}{\Thr}{\Thr}$.

The axioms for the projection operators and AIP are given in
Table~\ref{axioms-BTAinf}.
\begin{table}[!t]
\caption{Axioms of \BTAinf}
\label{axioms-BTAinf}
\begin{eqntbl}
\begin{axcol}
\proj{0}(x) = \DeadEnd                                  & \axiom{PR1} \\
\proj{n+1}(\DeadEnd) = \DeadEnd                         & \axiom{PR2} \\
\proj{n+1}(\Stop) = \Stop                               & \axiom{PR3} \\
\proj{n+1}(\pcc{x}{a}{y}) = \pcc{\proj{n}(x)}{a}{\proj{n}(y)}
                                                        & \axiom{PR4}
\eqnsep
\LAND{n \geq 0} \proj{n}(x) = \proj{n}(y) \Limpl x = y  & \axiom{AIP}
\end{axcol}
\end{eqntbl}
\end{table}
In this table, $a$ stands for an arbitrary basic action from $\BAct$ 
and $n$ stands for an arbitrary natural number from $\Nat$.
We write \BTAinf\ for \BTA\ extended with the projection operators, 
the axioms for the projection operators, and~AIP.

\section{Thread Extraction and Behavioural Congruence}
\label{sect-TE-BC}

In this section, we make precise in the setting of \BTAinf\ which 
behaviours are produced by \PGA\ instruction sequences under 
execution and introduce the notion of behavioural congruence on \PGA\ 
instruction sequences.

To make precise which behaviours are produced by \PGA\ instruction 
sequences under execution, we introduce an operator $\extr{\ph}$ meant 
for extracting from each \PGA\ instruction sequence the thread that 
models the behaviour produced by it under execution.
For each closed \PGAbc\ term $t$, $\extr{t}$ represents the thread that
models the behaviour produced by the instruction sequence represented 
by $t$ under execution.

Formally, we combine \PGAbc\ with \BTAinf\ and extend the 
combination with the \emph{thread extraction} operator 
$\funct{\extr{\ph}}{\InSeq}{\Thr}$ and 
the axioms given in Table~\ref{axioms-thread-extr}.%
\begin{table}[!t]
\caption{Axioms for the thread extraction operator}
\label{axioms-thread-extr}
\begin{eqntbl}
\begin{axcol}
\extr{a} = a \bapf \DeadEnd                            & \axiom{TE1}  \\
\extr{a \conc X} = a \bapf \extr{X}                    & \axiom{TE2}  \\
\extr{\ptst{a}} = a \bapf \DeadEnd                     & \axiom{TE3}  \\
\extr{\ptst{a} \conc X} = \pcc{\extr{X}}{a}{\extr{\fjmp{2} \conc X}}
                                                       & \axiom{TE4}  \\
\extr{\ntst{a}} = a \bapf \DeadEnd                     & \axiom{TE5}  \\
\extr{\ntst{a} \conc X} = \pcc{\extr{\fjmp{2} \conc X}}{a}{\extr{X}}
                                                       & \axiom{TE6}
\end{axcol}
\qquad
\begin{axcol}
\extr{\fjmp{l}} = \DeadEnd                             & \axiom{TE7}  \\
\extr{\fjmp{0} \conc X} = \DeadEnd                     & \axiom{TE8}  \\
\extr{\fjmp{1} \conc X} = \extr{X}                     & \axiom{TE9}  \\
\extr{\fjmp{l+2} \conc u} = \DeadEnd                   & \axiom{TE10} \\
\extr{\fjmp{l+2} \conc u \conc X} = \extr{\fjmp{l+1} \conc X}
                                                       & \axiom{TE11} \\
\extr{\halt} = \Stop                                   & \axiom{TE12} \\
\extr{\halt \conc X} = \Stop                           & \axiom{TE13}
\end{axcol}
\end{eqntbl}
\end{table}
In this table, 
$a$ stands for an arbitrary basic instruction from $\BInstr$, 
$u$ stands for an arbitrary primitive instruction from $\PInstr$, and 
$l$ stands for an arbitrary natural number from $\Nat$.

If a closed \PGAbc\ term $t$ represents an instruction sequence that
starts with an infinite chain of forward jumps, then TE9 and TE11 can 
be applied to $\extr{t}$ infinitely often without ever showing that a 
basic action is performed.
In this case, we have to do with inaction and, being consistent with 
that, $t = \fjmp{0} \conc t'$ is derivable from the axioms of \PGAsc\ 
for some closed \PGAbc\ term $t'$.
By contrast, $t = \fjmp{0} \conc t'$ is not derivable from the axioms of 
\PGAisc.
If closed \PGAbc\ terms $t$ and $t'$ represent instruction sequences in 
which no infinite chains of forward jumps occur, then $t = t'$ is 
derivable from the axioms of \PGAsc\ only if $\extr{t} = \extr{t'}$ is 
derivable from the axioms of \PGAisc\ and TE1--TE13.

If a closed \PGAbc\ term $t$ represents an infinite instruction 
sequence, then we can extract the approximations of the thread modeling
the behaviour produced by that instruction sequence under execution up 
to every finite depth: for each $n \in \Nat$, there exists a closed 
\BTA\ term $t''$ such that $\proj{n}(\extr{t}) = t''$ is derivable 
from the axioms of \PGAsc, TE1--TE13, the axioms of \BTA, and 
PR1--PR4.
If closed \PGAbc\ terms $t$ and $t'$ represent infinite instruction 
sequences that produce the same behaviour under execution, then this can
be proved using the following instance of AIP:
$\LAND{n \geq 0} \proj{n}(\extr{t}) = \proj{n}(\extr{t'}) \Limpl
 \extr{t} = \extr{t'}$.

\PGA\ instruction sequences are behaviourally equivalent if they 
produce the same behaviour under execution.
Behavioural equivalence is not a congruence.
Instruction sequences are behaviourally congruent if they produce the
same behaviour irrespective of the way they are entered and the way
they are left.

Let $t$ and $t'$ be closed \PGAbc\ terms.
Then:
\begin{itemize}
\item
$t$ and $t'$ \hsp{.35}are \emph{behaviourally equivalent}, 
\hsp{.175}written $t \beqv t'$, if $\extr{t} = \extr{t'}$ is derivable
from the axioms of \PGAsc, TE1--TE13, and the axioms of \BTAinf.
\item
$t$ and $t'$ are \emph{behaviourally congruent}, written 
$t \bcong t'$, if, for each $l,n \in \Nat$,
$\fjmp{l} \conc t \conc \halt^n \beqv \fjmp{l} \conc t' \conc \halt^n$.%
\footnote
{We use the convention that $t \conc {t'}^0$ stands for $t$.}
\end{itemize}
Behavioural congruence is the largest congruence contained in
behavioural equivalence.
Moreover, structural congruence implies behavioural congruence.
\begin{proposition}
\label{prop-scongr-bequiv}
For all closed \PGAbc\ terms $t$ and $t'$,
$t = t'$ is derivable from the axioms of \PGAsc\ only if $t \bcong t'$.
\end{proposition}
\begin{proof}
The proof is analogous to the proof of Proposition~2.2 
from~\cite{BM12b}.
In that proof use is made of the uniqueness of solutions of sets of 
recursion equations where each right-hand side is a \BTA\ term of
the form $\DeadEnd$, $\Stop$ or $\pcc{s}{a}{s'}$ with \BTA\ terms $s$ 
and $s'$ that contain only variables occurring as one of the right-hand 
sides.
This uniqueness follows from AIP (see also Corollary~2.1 
from~\cite{BM12b}). 
\qed 
\end{proof} 
Conversely, behavioural congruence does not implies structural
congruence.
For example,
$\ptst{a} \conc \halt \conc \halt \bcong
 \ntst{a} \conc \halt \conc \halt$,
but
$\ptst{a} \conc \halt \conc \halt =
 \ntst{a} \conc \halt \conc \halt$ 
is not derivable from the axioms of \PGAsc.

\section{Axioms of \PGAbc\ and Behavioural Congruence}
\label{sect-PGA-BC}

The axioms of \PGAbc\ are intended to be used for establishing
behavioural congruence in a direct way by nothing more than equational 
reasoning.
Two questions arise: the soundness question, i.e.\ the question whether 
derivable equality implies behavioural congruence, and the completeness 
question, i.e.\ the question whether behavioural congruence implies 
derivable equality.
The two theorems presented in this section concern these questions.
The first theorem fully answers the soundness question in the 
affirmative.
The second theorem answers the completeness question in the affirmative 
only for the restriction obtained by excluding the repetition operator
because of problems in mastering the intricacy of a completeness proof 
for the unrestricted case.   

We start with a few additional definitions and results which will be 
used in the proof of the theorems.

A closed \PGAbc\ term $t$ \emph{has simplifiable control flow} if there 
exists a closed \PGAbc\ term $t'$ such that $t = t'$ is derivable from 
the axioms of \PGAisc\ and $t'$ contains a subterm of the same form as 
the left-hand side of one of the axioms PGA9--PGA30. 
The intuition is that a closed \PGAbc\ term has simplifiable control 
flow if the instruction sequence that it represents has unnecessary 
tests, unnecessary jumps or needlessly long jumps.
A closed \PGAbc\ term is in \emph{third canonical form} if it is in 
second canonical form and does not have simplifiable control flow.

The following proposition relates \PGAbc\ and third canonical forms.
\begin{proposition}
\label{prop-3CF}
For all closed \PGAbc\ terms $t$, there exists a closed \PGAbc\ term 
$t'$ that is in third canonical form such that $t = t'$ is derivable 
from the axioms of \PGAbc.
\end{proposition}
\begin{proof}
By Proposition~\ref{prop-2CF}, there exists a closed 
\PGAbc\ term $t''$ that is in second canonical form such that $t = t''$ 
is derivable from the axioms of \PGAsc.
If $t''$ has simplifiable control flow, it can be transformed into a 
closed \PGAbc\ term that does not have simplifiable control flow by 
applications of PGA9--PGA30 possibly alternated with applications of 
PGA1 and/or PGA4. 
\qed 
\end{proof} 
Proposition~\ref{prop-3CF} is important to the proof of
Theorem~\ref{theorem-completeness} below.
Actually, there are some axioms among PGA9--PGA30 that did not turn up 
until the elaboration of the proof of Theorem~\ref{theorem-completeness}.

The set of \emph{basic \PGAbc\ terms} is inductively defined as follows:
\begin{itemize}
\item
if $u \in \PInstr$, then $u$ is a basic \PGAbc\ term;
\item
if $u \in \PInstr$ and $t$ is a basic \PGAbc\ term, then 
$u \conc t$ is a basic \PGAbc\ term; and
\item
if $t$ is a basic \PGAbc\ term, then $t\rep$ is a basic \PGAbc\ term.
\end{itemize}
Obviously, for all closed \PGAbc\ terms $t$, there exists a basic 
\PGAbc\ term $t'$ such that $t = t'$ is derivable from PGA1.

\begin{lemma}
\label{lemma-basic-3CF}
For all basic repetition-free \PGAbc\ terms $t$ that are in third 
canonical form,  $t$ is of one of the following forms:
\begin{enumerate}
\item[(a)]
$u$, \phantom{${}\! \conc t'$} where $u \in \PInstr$;
\item[(b)]
$u \conc t'$, where $u \in \PInstr$ and $t'$ is a basic repetition-free
\PGAbc\ term that is in third canonical form.
\end{enumerate}
\end{lemma}
\begin{proof}
This lemma with all occurrences of ``third canonical form'' replaced by 
``first canonical form'' follows immediately from the definitions of 
basic \PGAbc\ term and first canonical form.
Moreover, in the case that $t$ is of the form~(b), it follows 
immediately from the definitions concerned that the properties ``does 
not have chained jumps'',``has shortest possible jumps ending in the 
repeating part'', and ``does not have simplifiable control flow'' carry 
over from $t$ to $t'$.
This means that $t'$ is also in third canonical form. 
\qed 
\end{proof}
In the rest of this section, we refer to the possible forms of basic 
\PGAbc\ terms that are in third canonical form as in 
Lemma~\ref{lemma-basic-3CF}.
\begin{lemma}
\label{lemma-bcong-3CF}
For all basic repetition-free \PGAbc\ terms $t$ and $t'$ that are in 
third canonical form, $t \bcong t'$ only if
\begin{enumerate} 
\item[(1)]
$t$ is of the form~(a) iff $t'$ is of the form~(a);
\item[(2)]
$t$ is of the form~(b) iff $t'$ is of the form~(b).
\end{enumerate}
\end{lemma}
\begin{proof}
Suppose that $t$ and $t'$ are in third canonical form and $t \bcong t'$.

Property~(1) is trivial because, in the case that $t$ is of the 
form~(a), $t \bcong t'$ iff $t \equiv t'$.%
\footnote{We write $\equiv$ for syntactic identity.}

Property~(2) follows immediately from Lemma~\ref{lemma-basic-3CF} and 
the consequence of property~(1) that, in the case that $t$ is of the 
form~(b), $t'$ is not of the form~(a).
\qed 
\end{proof}

We now move on to the two theorems announced at the beginning of this 
section.

\begin{theorem}
\label{theorem-soundness}
For all closed \PGAbc\ terms $t$ and $t'$, $t = t'$ is derivable from the 
axioms of \PGAbc\ only if $t \bcong t'$.
\end{theorem}
\begin{proof}
Because $\bcong$ is a congruence, it is sufficient to prove for
each axiom of \PGAbc\ that, for all its closed substitution instances 
$t = t'$, \mbox{$t \bcong t'$}.
For PGA1--PGA8, this follows immediately from 
Proposition~\ref{prop-scongr-bequiv}.
For PGA9--PGA30, it follows very straightforwardly from the definition 
of $\bcong$, \mbox{TE1--TE13}, and in the case of PGA17 and PGA18, the 
unfolding equation $X\rep = X \conc X\rep$. 
\qed 
\end{proof}

\begin{theorem}
\label{theorem-completeness}
For all closed repetition-free \PGAbc\ terms $t$ and $t'$, $t = t'$ is 
derivable from the axioms of \PGAbc\ if $t \bcong t'$.
\end{theorem}
\begin{proof}
See Appendix~\ref{app-proof-theorem-2}. 
\qed 
\end{proof}

We will conclude Appendix~\ref{app-proof-theorem-2} by going into the 
main problem that we have experienced in mastering the intricacy of a 
proof of the unrestricted version of Theorem~\ref{theorem-completeness},
which reads as follows:
\begin{quote}
\em
for all closed \PGAbc\ terms $t$ and $t'$, $t = t'$ is derivable from 
the axioms of \PGAbc\ if $t \bcong t'$.
\end{quote}

\section{The Case of Instructions for Boolean Registers}
\label{sect-PGAbr}

In this section, we present the instantiation of \PGAbc\ in which all 
possible instructions for Boolean registers are taken as basic 
instructions.
This instantiation is called \PGAbrbc\ (\PGAbc\ with instructions for 
Boolean registers).
In order to justify the additional axioms of \PGAbrbc, we also present 
the instantiation of \BTA\ in which all possible instructions for 
Boolean registers are taken as basic actions and adapt the definitions 
of behavioural equivalence and behavioural congruence to closed 
\PGAbrbc\ terms using this instantiation of \BTA.

In \PGAbrbc, it is assumed that a fixed but arbitrary set $\Foci$ of 
\emph{foci} has been given.
Foci serve as names of Boolean register services.

The set of basic instructions used in \PGAbrbc\ consists of the 
following:
\begin{itemize}
\item
for each $f \in \Foci$ and $\funct{p,q}{\Bool}{\Bool}$,
a \emph{basic Boolean register instruction} $f.\mbr{p}{q}$.
\end{itemize}
We write $\BInstrbr$ for this set.

The intuition is that the execution of a basic Boolean register 
instruction may modify the register content of a Boolean register 
service and must produce a Boolean value as reply at its completion.
The actual reply may be dependent on the register content of the Boolean 
register service.
More precisely, the execution of a basic Boolean register instruction 
has the following effects:
\begin{itemize}
\item
if the register content of the Boolean register service named $f$ is $b$ 
when the execution of $f.\mbr{p}{q}$ starts, then its register content 
is $q(b)$ when the execution of $f.\mbr{p}{q}$ terminates;
\item
if the register content of the Boolean register service named $f$ is $b$ 
when the execution of $f.\mbr{p}{q}$ starts, then the reply produced on 
termination of the execution of $f.\mbr{p}{q}$ is $p(b)$.
\end{itemize}
The execution of $f.\mbr{p}{q}$ has no effect on the register content of 
Boolean register services other than the one named $f$.

$\Bool \to \Bool$, the set of all unary Boolean functions, consists of 
the following four functions:
\begin{itemize}
\item
the function $\FFunc$, satisfying 
$\FFunc(\False) = \False$ and $\FFunc(\True) = \False$;
\item
the function $\TFunc$, satisfying 
$\TFunc(\False) = \True$ and $\TFunc(\True) = \True$;
\item
the function $\IFunc$, satisfying 
$\IFunc(\False) = \False$ and $\IFunc(\True) = \True$;
\item
the function $\CFunc$, satisfying 
$\CFunc(\False) = \True$ and $\CFunc(\True) = \False$.
\end{itemize}
In~\cite{BM13a}, we actually used the methods $\mbr{\FFunc}{\FFunc}$, 
$\mbr{\TFunc}{\TFunc}$, and $\mbr{\IFunc}{\IFunc}$, but denoted them by 
$\setbr{0}$, $\setbr{1}$ and $\getbr$, respectively.
In~\cite{BM14e}, we actually used, in addition to these methods, the 
method $\mbr{\CFunc}{\CFunc}$, but denoted it by $\negbr$.

We write $\PInstrbr$ for the set $\PInstr$ of primitive instructions in 
the case where $\BInstrbr$ is taken as the set $\BInstr$.

The constants and operators of \PGAbrbc\ are the constants and operators 
of \PGAbc\ in the case where $\PInstrbr$ is taken as the set $\PInstr$.

Closed \PGAbrbc\ terms are considered equal if the instruction
sequences that they represent can always take each other's place in an
instruction sequence in the sense that the behaviour produced under 
execution remains the same under the intended interpretation of the 
instructions from $\BInstrbr$.
In other words, equality of closed terms stands in \PGAbrbc\ for a kind 
of behavioural congruence of the represented instruction sequences.
The kind of behavioural congruence in question will be made precise at 
the end of this section.

The axioms of \PGAbrbc\ are the axioms of \PGAbc\ and in addition the 
axioms given in Table~\ref{axioms-PGAbr}.%
\begin{table}[!t]
\caption{Additional axioms for \PGAbrbc}
\label{axioms-PGAbr}
\begin{eqntbl}
\begin{axcol}
\ptst{f.\mbr{\FFunc}{p}} = \ntst{f.\mbr{\TFunc}{p}}  & \axiom{PGAbr1} \\
\ptst{f.\mbr{\TFunc}{p}} = \ntst{f.\mbr{\FFunc}{p}}  & \axiom{PGAbr2} \\
\ptst{f.\mbr{\IFunc}{p}} = \ntst{f.\mbr{\CFunc}{p}}  & \axiom{PGAbr3} \\
\ptst{f.\mbr{\CFunc}{p}} = \ntst{f.\mbr{\IFunc}{p}}  & \axiom{PGAbr4} \\
\ptst{f.\mbr{\TFunc}{p}} = f.\mbr{q}{p}              & \axiom{PGAbr5} 
\end{axcol}
\end{eqntbl}
\end{table}
In this table, 
$f$ stands for an arbitrary focus from $\Foci$, and
$p$ and $q$ stand for arbitrary unary Boolean functions from 
$\Bool \to \Bool$.

If $t = t'$ is derivable from the axioms of \PGAbrbc, then $t$ and $t'$ 
represent instruction sequences that can always take each other's place
in an instruction sequence without affecting the behaviour produced 
under execution in an essential way, taking the intended interpretation 
of the instructions from $\BInstrbr$ into account.
Below, we introduce the instantiation of \BTA\ in which all possible 
instructions for Boolean registers are taken as basic actions to make 
this precise.

Henceforth, the instruction sequences of the kind considered in 
\PGAbrbc\ are called \PGAbr\ instruction sequences.

The instantiation of \BTA\ referred to above is called \BTAbr\ (\BTA\ 
with instructions for Boolean registers).
In \BTAbr, the effects of performing a basic action on both the register 
content of Boolean register services and the way in which the modeled 
behaviour proceeds subsequently to performing the basic action concerned
correspond to the intended interpretation of the basic action when it is
considered to be a basic instruction.

The constants and operators of \BTAbr\ are the constants and operators 
of \BTA\ in the case where $\BInstrbr$ is taken as the set $\BInstr$.

The idea behind equality of \BTAbr\ terms is that two closed \BTAbr\ 
terms are equal if they represent threads that can be made the same by 
a number of changes that never influences at any step of the modeled 
behaviour the effects of the basic action performed on the register 
content of Boolean register services and the way in which the modeled 
behaviour proceeds.
Equality of closed terms stands in \BTAbr\ for a kind of congruence of 
the represented threads which originates from the notion of effectual 
equivalence of basic instructions introduced in~\cite{BM15a}.

The axioms of \BTAbr\ are given in Table~\ref{axioms-BTAbr}.%
\begin{table}[!t]
\caption{Axioms of \BTAbr}
\label{axioms-BTAbr}
\begin{eqntbl}
\begin{axcol}
\pcc{x}{f.\mbr{\FFunc}{q}}{y} = \pcc{y}{f.\mbr{\TFunc}{q}}{x}
                                                     & \axiom{BTAbr1} \\
\pcc{x}{f.\mbr{\IFunc}{q}}{y} = \pcc{y}{f.\mbr{\CFunc}{q}}{x}
                                                     & \axiom{BTAbr2} \\
\pcc{x}{f.\mbr{\TFunc}{q}}{y} = \pcc{x}{f.\mbr{p}{q}}{x}
                                                     & \axiom{BTAbr3} 
\end{axcol}
\end{eqntbl}
\end{table}
In this table, 
$f$ stands for an arbitrary focus from $\Foci$, and
$p$ and $q$ stand for arbitrary unary Boolean functions from 
$\Bool \to \Bool$.

Like \BTA, we can extend \BTAbr\ with the projection operators, the 
axioms for the projection operators and AIP.
We write \BTAbrinf\ for the resulting theory.

To make precise which behaviours are produced by \PGAbr\ instruction 
sequences under execution, we combine \PGAbrbc\ with \BTAbrinf\ and 
extend the combination with the thread extraction operator and the 
axioms for the thread extraction operator.

\PGAbr\ instruction sequences are behaviourally equivalent if the 
behaviours that they produce under execution are the same under the 
intended interpretation of the instructions from $\BInstrbr$.

Let $t_1$ and $t_2$ be closed \PGAbrbc\ terms.
Then:
\begin{itemize}
\item
$t$ and $t'$ \hsp{.35}are \emph{behaviourally equivalent}, 
\hsp{.175}written $t \beqv t'$, if $\extr{t} = \extr{t'}$ is derivable
from the axioms of \PGAsc, TE1--TE13, and the axioms of \BTAbrinf.
\item
$t$ and $t'$ are \emph{behaviourally congruent}, written 
$t \bcong t'$, if, for each $l,n \in \Nat$,
$\fjmp{l} \conc t \conc \halt^n \beqv \fjmp{l} \conc t' \conc \halt^n$.
\end{itemize}

It is obvious that,  with this adapted definition of behavioural 
congruence, Theorem~\ref{theorem-soundness} goes through for closed 
\PGAbrbc\ terms and Theorem~\ref{theorem-completeness} goes through for 
closed repetition-free \PGAbrbc\ terms.

\section{Concluding Remarks}
\label{sect-concl}

In program algebra, three congruences on instruction sequences are 
paid attention to: instruction sequence congruence, structural 
congruence, and behavioural congruence.
However, an axiom system for behavioural congruence had never been 
given.
In this paper, we have given an axiom system for behavioural congruence 
and proved its soundness for closed terms and completeness for closed
repetition-free terms.
This means that behavioural congruence of finite instruction sequences
can now be established in a direct way by nothing more than equational 
reasoning.
In earlier work, it had to be established in an indirect way, namely via 
thread extraction, by reasoning that was not purely equational.  
It is an open question whether the axiom system is also complete for 
closed terms in the case where the closed terms considered are not 
restricted to the repetition-free ones.

\appendix

\section{Appendix}
\label{app-proof-theorem-2}

In this appendix, we outline the proof of 
Theorem~\ref{theorem-completeness}.
We do not give full details of the proof because the full proof is 
really tedious.
We have aimed at providing sufficient information in the outline of the
proof to make a reconstruction of the full proof a routine matter.

\subsubsection*{Proof of Theorem~\protect\ref{theorem-completeness}:}
For all closed \PGAbc\ terms $s$, there exists a basic \PGAbc\ term $s'$ 
such that $s = s'$ is derivable from PGA1.
Moreover, for all closed \PGAbc\ terms $s$ and $s'$, $s = s'$ is 
trivially derivable from the axioms of \PGAbc\ if $s \equiv s'$.
By these facts, Proposition~\ref{prop-3CF}, and 
Theorem~\ref{theorem-soundness}, it is sufficient to prove:
\begin{quote}
for all basic repetition-free \PGAbc\ terms $t$ and $t'$ that are in 
third canonical form, $t \equiv t'$ if $t \bcong t'$.
\end{quote}
We prove this by induction on the depth of $t$ and case distinction on 
the form of $t$ according to Lemma~\ref{lemma-basic-3CF}.

The case $t \equiv u$, for $u \in \PInstr$, is trivial because 
$t \bcong t'$ only if $t \equiv t'$.

The case $t \equiv u \conc s$, for $u \in \PInstr$ and basic 
repetition-free \PGAbc\ term $s$ that is in third canonical form, is 
more involved.
It follows immediately from Lemma~\ref{lemma-bcong-3CF} that in this 
case $t \bcong t'$ only if $t' \equiv u' \conc s'$ for some 
$u' \in \PInstr$ and basic repetition-free \PGAbc\ term $s'$ that is in 
third canonical form.
Let $u' \in \PInstr$ and $s'$ be a basic repetition-free \PGAbc\ term 
that is in third canonical form such that $t' \equiv u' \conc s'$.
Then it follows immediately from the definition of $\bcong$ that 
$t \bcong t'$ only if $s \bcong s'$.
Hence, by the induction hypothesis, we have that $t \bcong t'$ only if 
$s \equiv s'$.
We proceed with a case analysis on $(u,u')$.
There exist 25 combinations of kinds of primitive instructions.
In 9 of these combinations, it matters whether the basic instructions
involved are the same and, in 1 of these combinations, it matters 
whether the natural numbers involved are the same.
Hence, in total, there are 35 cases to consider.
However, 5 cases are trivial because in those cases $u \equiv u'$ and
13 cases are covered by a symmetric case.
Of the remaining 17 cases, 9 cases contradict $t \bcong t'$.
Left over are the following 8 cases:
$(u,u') = (\ptst{a},a)$, $(u,u') = (\ntst{a},a)$, 
$(u,u') = (\ptst{a},\ntst{a})$, 
$(u,u') = (\fjmp{l},\ptst{a})$, $(u,u') = (\fjmp{l},\ntst{a})$, 
$(u,u') = (\fjmp{l},a)$, 
$(u,u') = (\fjmp{l},\fjmp{l'})$ with $l \neq l'$, 
$(u,u') = (\fjmp{l},\halt)$.
The proof now continues with a case analysis on $(s,s')$ for each of 
these eight cases, using implicitly the above-mentioned fact that 
$s \equiv s'$ each time that the conclusion is drawn that there is a 
contradiction with $t \bcong t'$.
We will also implicitly use several times the easy to check fact that,
for all basic repetition-free \PGAbc\ terms $r$ that are in third 
canonical form, $\extr{r} \neq \extr{\fjmp{l{+}2} \conc r}$ and
$\extr{r} \neq \extr{u_1 \conc \ldots \conc u_{k+1} \conc r}$ if
$u_1 \equiv a$ or $u_1 \equiv \ptst{a}$ or $u_1 \equiv \ntst{a}$.

In the analysis for the case $(u,u') = (\ptst{a},a)$, we make a case 
distinction on the form of $s$ according to Lemma~\ref{lemma-basic-3CF}:
\begin{itemize}
\item
in the case that $s \equiv v$, we make a further case distinction on the 
form of $v$:
\begin{itemize}
\item
if $v \equiv b$ or $v \equiv \ptst{b}$ or $v \equiv \ntst{b}$, then we 
have $\extr{\ptst{a} \conc v} \neq \extr{a \conc v}$
and hence a contradiction with $t \bcong t'$;
\item
if $v \equiv \fjmp{0}$, then we have
$\extr{\ptst{a} \conc v \conc \halt} \neq \extr{a \conc v \conc \halt}$
and hence a contradiction with $t \bcong t'$;
\item
if $v \equiv \fjmp{1}$, then $t$ is not in third canonical form;
\item
if $v \equiv \fjmp{l{+}2}$, then we have
$\extr{\ptst{a} \conc v \conc \halt} \neq \extr{a \conc v \conc \halt}$
and hence a contradiction with $t \bcong t'$;
\item
if $v \equiv \halt$, then we have
$\extr{\ptst{a} \conc v} \neq \extr{a \conc v}$ 
and hence a contradiction with $t \bcong t'$;
\end{itemize}
\item
in the case that $s \equiv v \conc r$, for some basic repetition-free 
\PGAbc\ term $r$ that is in third canonical form, we make a further case 
distinction on the form of $v$ as well:
\begin{itemize}
\item
if $v \equiv b$ or $v \equiv \ptst{b}$ or $v \equiv \ntst{b}$, then we
have $\extr{\ptst{a} \conc v \conc r} \neq \extr{a \conc v \conc r}$, 
because $r$ is repetition-free, and hence a contradiction with 
$t \bcong t'$;
\item
if $v \equiv \fjmp{0}$, then it follows from $t \bcong t'$ that  
$r \equiv \fjmp{0}$ or $r \equiv \fjmp{0} \conc r'$ for some $r'$ and 
hence $t$ is not in third canonical form;
\item
if $v \equiv \fjmp{1}$, then $t$ is not in third canonical form;
\item
if $v \equiv \fjmp{l{+}2}$, then we make a further case distinction on 
the form of $r$ according to Lemma~\ref{lemma-basic-3CF}:
\begin{itemize}
\item
in the case that $r \equiv w$, we make a further case distinction on the 
form of $w$:
\begin{itemize}
\item
if $w \equiv b$ or $w \equiv \ptst{b}$ or $w \equiv \ntst{b}$, then we
have 
$\extr{\ptst{a} \conc \fjmp{l{+}2} \conc w} \neq 
 \extr{a \conc \fjmp{l{+}2} \conc w}$ 
and hence a contradiction with $t \bcong t'$;
\item
if $w \equiv \fjmp{0}$, then we have 
$\extr{\ptst{a} \conc \fjmp{l{+}2} \conc w \conc \halt^{l{+}1}} \neq 
 \extr{a \conc \fjmp{l{+}2} \conc w \conc \halt^{l{+}1}}$ 
and hence a contradiction with $t \bcong t'$;
\item
if $w \equiv \fjmp{l'{+}1}$ and $l' > l$, then we have
$\extr{\ptst{a} \conc \fjmp{l{+}2} \conc w \conc \halt^{l{+}1}} \neq 
 \extr{a \conc \fjmp{l{+}2} \conc w \conc \halt^{l{+}1}}$ 
and hence a contradiction with $t \bcong t'$;
\item
if $w \equiv \fjmp{l'{+}1}$ and $l' < l$, then we have
$\extr{\ptst{a} \conc \fjmp{l{+}2} \conc w \conc \halt^{l'{+}1}} \neq 
 \extr{a \conc \fjmp{l{+}2} \conc w \conc \halt^{l'{+}1}}$ 
and hence a contradiction with $t \bcong t'$;
\item
if $w \equiv \fjmp{l'{+}1}$ and $l' = l$, then $t$ is not in third 
canonical form;
\item
if $w \equiv \halt$, then we have 
$\extr{\ptst{a} \conc \fjmp{l{+}2} \conc w} \neq 
 \extr{a \conc \fjmp{l{+}2} \conc w}$ 
and hence a contradiction with $t \bcong t'$;
\end{itemize}
\item
in the case that $r \equiv w \conc r'$, for some basic repetition-free 
\PGAbc\ term $r'$ that is in third canonical form, we make a further 
case distinction on the form of $w$ as well:
\begin{itemize}
\item
if $w \equiv b$ or $w \equiv \ptst{b}$ or $w \equiv \ntst{b}$, then we
have 
$\extr{\ptst{a} \conc \fjmp{l{+}2} \conc w \conc r'} \neq 
 \extr{a \conc \fjmp{l{+}2} \conc w \conc r'}$, 
because $r'$ is repetition-free, and hence a con\-tradiction with 
$t \bcong t'$;
\item
if $w \equiv \fjmp{0}$, then it follows from $t \bcong t'$ that  
$r' \equiv w_1 \conc \ldots \conc w_l \conc \fjmp{0}$ or 
$r' \equiv w_1 \conc \ldots \conc w_l \conc \fjmp{0} \conc r''$ for some 
$r''$ and hence $t$ is not in third canonical form;
\item
if $w \equiv \fjmp{l'{+}1}$ and $l' \neq l$, then we
have 
$\extr{\ptst{a} \conc \fjmp{l{+}2} \conc w \conc r'} \neq 
 \extr{a \conc \fjmp{l{+}2} \conc w \conc r'}$, 
because $r'$ is repetition-free, and hence a con\-tradiction with 
$t \bcong t'$;
\item
if $w \equiv \fjmp{l'{+}1}$ and $l' = l$, then $t$ is not in third 
canonical form; 
\item
if $w \equiv \halt$, then it follows from $t \bcong t'$ that  
$r' \equiv w_1 \conc \ldots \conc w_l \conc \halt$ or 
$r' \equiv w_1 \conc \ldots \conc w_l \conc \halt \conc r''$ for some 
$r''$ and hence $t$ is not in third canonical form;
\end{itemize}
\end{itemize}
\item
if $v \equiv \halt$, then it follows from $t \bcong t'$ that  
$r \equiv \halt$ or $r \equiv \halt \conc r'$ for some $r'$ and hence
$t$ is not in third canonical form.
\end{itemize}
\end{itemize}
We conclude from this analysis that, in the case that 
$t \equiv \ptst{a} \conc s$ and $t' \equiv a \conc s$ for some basic 
repetition-free \PGAbc\ term $s$ that is in third canonical form, we 
have a contradiction with $t \bcong t'$.

The analyses for the cases $(u,u') = (\ntst{a},a)$ and 
$(u,u') = (\ptst{a},\ntst{a})$ are similar to the analysis for the case 
$(u,u') = (\ptst{a},a)$.

In the analysis for the case $(u,u') = (\fjmp{l},a)$, we make a case 
distinction on $l$:
\begin{itemize}
\item
if $l = 0$, then we have 
$\extr{\fjmp{l} \conc s} \neq \extr{a \conc s}$
and hence a contradiction with $t \bcong t'$;
\item
if $l = 1$, then we have 
$\extr{\fjmp{l} \conc s} \neq \extr{a \conc s}$,
because $s$ is repetition-free, and hence a contradiction with 
$t \bcong t'$;
\item
if $l = l' + 2$, then we make a further case distinction on the form of 
$s$ according to Lemma~\ref{lemma-basic-3CF}:
\begin{itemize}
\item
in the case that $s \equiv v$, we have 
$\extr{\fjmp{l} \conc s} \neq \extr{a \conc s}$
and hence a contradiction with $t \bcong t'$;
\item
in the case that $s \equiv v \conc r$, for some basic repetition-free 
\PGAbc\ term $r$ that is in third canonical form, it follows from 
$t \bcong t'$ that 
$r \equiv v_1 \conc \ldots \conc v_{l'} \conc a \conc r'$ 
for some basic repetition-free \PGAbc\ term $r'$ that is in third 
canonical form and we make a further case distinction on the form of 
$v$:
\begin{itemize}
\item
if $v \equiv b$ or $v \equiv \ptst{b}$ or $v \equiv \ntst{b}$, then we
have 
$\extr{\fjmp{l'{+}2} \conc v \conc v_1 \conc \ldots \conc v_{l'} \conc
 a \conc r'} \neq
 \extr{a \conc v \conc v_1 \conc \ldots \conc v_{l'} \conc a \conc r'}$, 
because $r'$ is repetition-free, and hence a contradiction with 
$t \bcong t'$;
\item
if $v \equiv \fjmp{0}$, then it follows from $t \bcong t'$ that  
$r' \equiv \fjmp{0}$ or $r' \equiv \fjmp{0} \conc r''$ for some $r''$ 
and hence $t$ is not in third canonical form;
\item
if $v \equiv \fjmp{l''{+}1}$ and $l'' \neq l' + 1$, then we have 
$\extr{\fjmp{l'{+}2} \conc v \conc v_1 \conc \ldots \conc v_{l'} \conc
       a \conc r'} \neq
 \extr{a \conc v \conc v_1 \conc \ldots \conc v_{l'} \conc a \conc r'}$, 
because $r'$ is repetition-free, and hence a contradiction with 
$t \bcong t'$;
\item
if $v \equiv \fjmp{l''{+}1}$ and $l'' = l' + 1$, then $t$ is not in 
third canonical form;
\item
if $v \equiv \halt$, then it follows from $t \bcong t'$ that  
$r' \equiv \halt$ or $r' \equiv \halt \conc r''$ for some $r''$ and 
hence $t$ is not in third canonical form.
\end{itemize}
\end{itemize}
\end{itemize}
We conclude from this analysis that, in the case that
$t \equiv \fjmp{l} \conc s$ and $t' \equiv a \conc s$ for some basic 
repetition-free \PGAbc\ term $s$ that is in third canonical form, we 
have a contradiction with $t \bcong t'$.

The analyses for the cases $(u,u') = (\fjmp{l},\ptst{a})$ and 
$(u,u') = (\fjmp{l},\ntst{a})$ are similar to the analysis for the case 
$(u,u') = (\fjmp{l},a)$.

In the analyses for the cases $(u,u') = (\fjmp{l},\fjmp{l'})$, with 
$l \neq l'$, and $(u,u') = (\fjmp{l},\halt)$, we use the function 
$\len$, which assigns to each closed repetition-free \PGAbc\ term the 
length of the instruction sequence that it represents.
This function is recursively defined as follows: $\len(u) = 1$ and 
$\len(t \conc t') = \len(t) + \len(t')$. 

In the analysis for the case $(u,u') = (\fjmp{l},\fjmp{l'})$ with 
$l \neq l'$, we only consider the case $l < l'$ (because the cases 
$l < l'$ and $l > l'$ are symmetric) and make a case distinction on $l$:
\begin{itemize}
\item
if $l = 0$, then we have 
$\extr{\fjmp{0} \conc s} \neq \extr{\fjmp{l'} \conc s}$
and hence a contradiction with $t \bcong t'$;
\item
if $0 < l \leq \len(s)$, then it follows from $t \bcong t'$ that 
$s \equiv 
 u_1 \conc \ldots \conc u_{l{-}1} \conc a \conc
 v_1 \conc \ldots \conc v_{l'{-}(l{+}1)} \conc a \conc r$ 
for some basic repetition-free \PGAbc\ term $r$ that is in third 
canonical form and we make a further case distinction on the form 
of $v$:
\begin{itemize}
\item
if $v_1 \equiv b$ or $v_1 \equiv \ptst{b}$ or $v_1 \equiv \ntst{b}$, 
then we have 
$\extr{\fjmp{l} \conc u_1 \conc \ldots \conc u_{l{-}1} \conc a \conc
       \linebreak[2]
       v_1 \conc \ldots \conc v_{l'{-}(l{+}1)} \conc a \conc r} \neq
 \extr{\fjmp{l'} \conc u_1 \conc \ldots \conc u_{l{-}1} \conc a \conc
       v_1 \conc \ldots \conc v_{l'{-}(l{+}1)} \conc a \conc r}$, 
because $r$ is repetition-free, and hence a contradiction with 
$t \bcong t'$;
\item
if $v_1 \equiv \fjmp{0}$, then it follows from $t \bcong t'$ that  
$r \equiv \fjmp{0}$ or $r \equiv \fjmp{0} \conc r'$ for some $r'$ 
and hence $t$ is not in third canonical form;
\item
if $v_1 \equiv \fjmp{l''{+}1}$ and $l'' \neq l' - l$, then we have
$\extr{\fjmp{l} \conc u_1 \conc \ldots \conc u_{l{-}1} \conc a \conc
       \linebreak[2]
       v_1 \conc \ldots \conc v_{l'{-}(l{+}1)} \conc a \conc r} \neq
 \extr{\fjmp{l'} \conc u_1 \conc \ldots \conc u_{l{-}1} \conc a \conc
       v_1 \conc \ldots \conc v_{l'{-}(l{+}1)} \conc a \conc r}$, 
because $r$ is repetition-free, and hence a contradiction with 
$t \bcong t'$;
\item
if $v_1 \equiv \fjmp{l''{+}1}$ and $l'' = l' - l$, then $t$ is not in 
third canonical form;
\item
if $v_1 \equiv \halt$, then it follows from $t \bcong t'$ that  
$r \equiv \halt$ or $r \equiv \halt \conc r'$ for some $r'$ 
and hence $t$ is not in third canonical form;
\end{itemize}
\item
if $l > \len(s)$, then we have 
$\extr{\fjmp{0} \conc s \conc \halt^{l{-}\len(s)}} \neq
 \extr{\fjmp{l'} \conc s \conc \halt^{l{-}\len(s)}}$
and hence a contradiction with $t \bcong t'$.
\end{itemize}
We conclude from this analysis that, in the case that
$t \equiv \fjmp{l} \conc s$ and $t' \equiv \fjmp{l'} \conc s$, with 
$l \neq l'$, for some basic repetition-free \PGAbc\ term $s$ that is in 
third canonical form, we have a contradiction with $t \bcong t'$.

In the analysis for the case $(u,u') = (\fjmp{l},\halt)$, we make a case 
distinction on $l$:
\begin{itemize}
\item
if $l = 0$, then we have 
$\extr{\fjmp{l} \conc s} \neq \extr{\halt \conc s}$
and hence a contradiction with $t \bcong t'$;
\item
if $0 < l \leq \len(s)$, then it follows from $t \bcong t'$ that 
$s \equiv u_1 \conc \ldots \conc u_{l{-}1} \conc \halt$ 
or
$s \equiv u_1 \conc \ldots \conc u_{l{-}1} \conc \halt \conc r$ 
for some $r$ and hence $t$ is not in third canonical form;
\item
if $l > \len(s)$, then we have 
$\extr{\fjmp{l} \conc s} \neq \extr{\halt \conc s}$
and hence a contradiction with $t \bcong t'$.
\end{itemize}
We conclude from this analysis that, in the case that
$t \equiv \fjmp{l} \conc s$ and $t' \equiv \halt \conc s$ for some basic 
repetition-free \PGAbc\ term $s$ that is in third canonical form, we 
have a contradiction with $t \bcong t'$.

From the conclusions of the analyses, it follows immediately that for 
all basic repetition-free \PGAbc\ terms $t$ and $t'$ that are in third 
canonical form, $t \equiv t'$ if $t \bcong t'$.
\qed

We conclude this appendix by going into the main problem that we have 
experienced in mastering the intricacy of a proof of the generalization 
of Theorem~\ref{theorem-completeness} from all closed repetition-free 
\PGAbc\ terms to all closed \PGAbc\ terms.

In the proof of Theorem~\ref{theorem-completeness}, case distinctions 
are made on a large scale.
It frequently occurs that the number of cases to be distinguished is 
kept small by making use of Lemma~\ref{lemma-basic-3CF}.
To devise and prove a generalization of this lemma that is not 
restricted to repetition-free terms is not a big problem.
In the proof of Theorem~\ref{theorem-completeness}, something of the 
following form occurs at many places: 
``we have $\extr{s} \neq \extr{s'}$ because $r$ is repetition-free, 
and hence a contradiction with $t \bcong t'$''.
At several similar places in the proof of the generalization of this 
theorem, $r$ is not repetition-free and $\extr{s} \neq \extr{s'}$ 
requires an elaborate proof. 
In some of these proofs, no use can be made of the generalization of 
Lemma~\ref{lemma-basic-3CF} and one gets completely lost in the many 
deeply nested case distinctions.
This is the main problem that we have experienced.

\bibliography{IS}

\end{document}